\documentclass[10pt,journal,compsoc]{IEEEtran}

\usepackage{booktabs} % For formal tables
\usepackage{multirow}
\usepackage[utf8]{inputenc}

\usepackage[ruled]{algorithm2e} % For algorithms
\renewcommand{\algorithmcfname}{ALGORITHM}
\SetAlFnt{\small}
\SetAlCapFnt{\small}
\SetAlCapNameFnt{\small}
\SetAlCapHSkip{0pt}
\IncMargin{-\parindent}

\usepackage{amsmath}
\usepackage{breqn}
\usepackage{txfonts}
\usepackage{graphicx}
\usepackage{colortbl}
\usepackage{bbding}
\usepackage{epstopdf}
\usepackage{cite}
\usepackage{xcolor}
\usepackage{graphicx}
\usepackage{url}
\usepackage[margin=0.6in]{geometry}

\usepackage{graphicx}
\usepackage{amssymb}
\usepackage{amsfonts}
\usepackage{cite}
\usepackage{cases}
\usepackage{txfonts}
\usepackage{subfigure}
\usepackage{threeparttable}
\usepackage{mathrsfs}
\usepackage{amsbsy}
\usepackage{booktabs}
\usepackage{multirow}
\usepackage{enumerate}
\usepackage{stfloats}
\usepackage{array}
\usepackage{xcolor}

\newtheorem{theorem}{Theorem}[section]

\newtheorem{lemma}[theorem]{Lemma}

\newtheorem{observation}{Observation}

\newenvironment{proof}[1][Proof]{\begin{trivlist}
\item[\hskip \labelsep {\bfseries #1}]}{\end{trivlist}}

\newenvironment{remark}[1][Remark]{\begin{trivlist}
\item[\hskip \labelsep {\bfseries #1}]}{\end{trivlist}}
\newcommand{\qed}{\nobreak \ifvmode \relax \else
     \ifdim\lastskip<1.5em \hskip-\lastskip
     \hskip1.5em plus0em minus0.5em \fi \nobreak
     \vrule height0.75em width0.5em depth0.25em\fi}

\begin{document}

\title{Improved Fault Analysis on SIMECK Ciphers}

\author{Duc-Phong~Le, %~\IEEEmembership{Member,~IEEE,}
        Rongxing~Lu,~\IEEEmembership{Senior~Member,~IEEE,}
        and~Ali~A.~Ghorbani,~\IEEEmembership{Senior~Member,~IEEE}% <-this % stops a space
\IEEEcompsocitemizethanks{\IEEEcompsocthanksitem D.-P.~Le, R. Lu, and A. Ghorbani are with the Canadian Institute for Cybersecurity,  Faculty of Computer Science, University of New Brunswick, Fredericton, Canada E3B 5A3. e-mail: (le.duc.phong@unb.ca, rlu1@unb.ca, ghorbani@unb.ca).\protect\\
}% <-this % stops an unwanted space
\thanks{Manuscript received April 19, 2005; revised August 26, 2015.}}

\markboth{Journal of \LaTeX\ Class Files,~Vol.~14, No.~8, August~2015}%
{Shell \MakeLowercase{\textit{et al.}}: Bare Demo of IEEEtran.cls for Computer Society Journals}

\IEEEtitleabstractindextext{%
\begin{abstract}
The advances of Internet of Things (IoT) have had a fundamental impact and influence in sharping our rich living experiences. However, since IoT devices are usually resource-constrained, lightweight block ciphers have played a major role in serving as a building block for secure IoT protocols.  In CHES 2015, SIMECK, a family of block ciphers, was designed for resource-constrained IoT devices.
Since its publication, there have been many analyses on its security.
In this paper, under the one bit-flip model, we  propose a new efficient fault analysis attack on SIMECK ciphers. Compared to those previously reported attacks, our attack can recover the full master key by injecting faults into only a {\em single round} of all SIMECK family members. This property is crucial, as it is infeasible for an attacker to inject faults into different rounds of a SIMECK implementation on IoT devices in the real world. Specifically, our attack is characterized by exercising a deep analysis of differential trail between the correct and faulty immediate ciphertexts. Extensive simulation evaluations are conducted, and the results demonstrate the effectiveness and correctness of our proposed attack.
\end{abstract}

}

\maketitle

\IEEEdisplaynontitleabstractindextext
\IEEEpeerreviewmaketitle

\IEEEraisesectionheading{\section{Introduction}\label{sec:introduction}}

\IEEEPARstart{O}{ne} of the key technologies transforming our living experiences into much smarter ones is the Internet of Things (IoT).  However, since   IoT devices are usually resource-constrained, it is pressing to design lightweight block ciphers as building blocks to secure IoT protocols, \emph{i.e.}, providing basic security requirements including confidentiality, data integrity, and source authentication for IoT devices.

By now, there have been numerous lightweight block ciphers introduced. A statistic performed by researchers of the University of Luxembourg~\footnote{https://www.cryptolux.org/index.php/Lightweight_Block_Ciphers} lists 36 lightweight block ciphers proposed by 2016. 
In order to evaluate and standardize lightweight cryptographic algorithms, NIST is supporting a Lightweight Cryptography project. By now, 32 candidates have been selected and are evaluating in Round 2. Readers can find more information about the project in this link \url{https://csrc.nist.gov/projects/lightweight-cryptography}.

In CHES 2015,  Yang {\it et al.}~\cite{simeck} introduced SIMECK, which is a family of lightweight block ciphers based upon Feistel structure and combines the good design principles of the SIMON and SPECK block ciphers~\cite{BSS+13}. Precisely, SIMECK  consists of three members with block sizes of $32, 48$, and $64$, and the corresponding key sizes are $64, 96$, and $128$, respectively. As demonstrated in~\cite{simeck}, SIMECK allows a smaller and more efficient hardware implementation in comparison to  SIMON. Due to its nice property in efficiency, SIMECK has been significantly analyzed since its publication.

As is known to all, fault analysis is a very efficient implementation attack against cryptographic protocols, which essentially tries to influence the behavior of a cryptographic device and determine sensitive information by examining the effects. Today, there have been several mechanisms to inject faults into microprocessors. Examples include changes in the power supply, the clock frequency~\cite{KK99}, or intensive lighting of the circuit~\cite{SA02}. In most cases, injecting faults will force a change in the data located in one of the registers.
The first fault attack against RSA-CRT implementation was reported in {a Bellcore press in 1996} and subsequently analyzed by Boneh, DeMillo and Lipton in~\cite{BDL97}. Concretely, Boneh \emph{et al.}~\cite{BDL97} showed that many implementations of RSA signatures and other public-key algorithms are vulnerable to a certain transient fault occurring during the processing phase, and the RSA-CRT implementation is at extreme risk to be compromised by using one erroneous result. After that, Biham and Shamir ~\cite{BS97} introduced the Differential Fault Analysis (DFA) against symmetric cryptosystems such as the DES~\cite{BS97}. They assume that an attacker can disturb DES computations by using the same - but {\it unknown} - plaintext at the last (three) DES round(s). The wrong ciphers provide a system of equations for the unknown last round key bits that finally reveals the correct key value.
Since then, there have been many DFA attacks carried out on other block ciphers, including attacks against AES~\cite{Gir04,kim2011improved}, Triple DES~\cite{HJQ04}, IDEA~\cite{CGV08}, SIMON and SPECK~\cite{tupsamudre2014differential, LYK19} or KLEIN lightweight block ciphers~\cite{GruberS19}. These attacks are performed on the key schedule~\cite{kim2011improved}, S-box~\cite{GruberS19}, or intermediate inputs~\cite{tupsamudre2014differential}. They are also carried out under various fault models (see more details in Section~2.3.2).

In 2016, Nalla \emph{et al.} ~\cite{NSS16} presented two fault analysis attacks against SIMECK ciphers. While the former is under the one bit-flip fault model, the latter is under the random byte fault model. Both of the two attacks aim at recovering the last round key by injecting faults into the second last round.
 In this paper, we would like to present an improved fault analysis attack against SIMECK block ciphers under the one bit-flip model, which is more practical than Nalla \emph{et al.}'s attacks. Specifically, the main contribution of this paper is three-fold.
 \begin{itemize}
   \item First,  we show  that the whole master key of SIMECK block ciphers could be recovered by injecting faults into a single round of the ciphers, which makes our attack more {\em practical} than  Nalla \emph{et al.}'s attacks~\cite{NSS16}, as their attacks require faults from $4$ different rounds.
   \item Second, by deducing more key bits with one fault, our attack also requires fewer faults than those previously reported attacks.
   \item  Third, we conduct extensive simulation evaluations, and the results demonstrate the effectiveness and correctness of our attack.
 \end{itemize}

The remainder of this paper is organized as follows. Section~\ref{background} briefly recalls SIMECK ciphers, security analyses, and differential fault analysis on this family of ciphers. 
Section~\ref{sec:observations} discusses some facts and observations that will be used in our attack.
Section~\ref{sec:bitflip} describes our differential fault analysis under the one bit-flip model. Then, we present our simulation evaluations in Section~\ref{sec:simulation}. Finally, we draw our conclusions in Section~\ref{sec:conclusion}.

\section{Preliminaries}\label{background}
In this section, we briefly recall the specification of the SIMECK family of lightweight block ciphers and differential fault analysis on this family of ciphers.
%\subsection{Notations.}
For the sake of consistency, we make use of the following notations listed in Table~\ref{tab:notations} in the rest of the paper. In principles, we use capital letters for vectors or strings while small letters are used to represent individual bits.

\begin{table}[ht]
\centering
\caption{Definition of Notations}
\label{tab:notations}
\footnotesize
\begin{tabular}{p{0.07\textwidth}p{ 0.36\textwidth}}

\toprule
~Notation& ~Definition\\
\midrule \midrule

\rowcolor[gray]{.9}  $T$ &    the total number of rounds of the cipher.\\
$(X^i, Y^i)$ &   input of round $i$, $i=0,1,\ldots,T-1$. \\
\rowcolor[gray]{.9} $(X^T, Y^T)$ &   ciphertext. \\
$(\hat{X}^{ i}, \hat{Y}^{ i})$ &   faulty input of round $i$, $i=0,1,\ldots,T-1$. \\
\rowcolor[gray]{.9} $(\hat{X}^{ T}, \hat{Y}^{ T})$ &  faulty ciphertext. \\
$X^i + Y^i$ &  `XOR' operation of $X^i$ and $Y^i$.\\
\rowcolor[gray]{.9} $X^i Y^i$ &  `AND' operation of $X^i$ and $Y^i$.\\
$(x^i_j, y^i_j)$ &  variables denoting bit $j$ of the input of the round $i$, $i=0,1,\ldots,T$, $ \quad j=0,1,\ldots,n-1$.\\
\rowcolor[gray]{.9} $(\hat{x}^{i}_j, \hat{y}^{i}_j)$ &   variables denoting bit $j$ of the faulty input of the round $i$, $i=0,1,\ldots,T$, $ \quad j=0,1,\ldots,n-1$.\\
$x^i_{j} + y^i_{j}$ &  bitwise 'XOR' operation of $x^i_j$ and $y^i_j$. \\
\rowcolor[gray]{.9} $x^i_{j} y^i_{j}$ & bitwise 'AND' operation of $x^i_j$ and $y^i_j$. \\
$K^i$ &  the round key in round $i$, $i=0,1,\ldots,T-1$. \\
\rowcolor[gray]{.9} $k^i_{j}$ &   $j^{th}, 0 \le j \le n-1$ bit round key of round $i$.\\
$\Delta^i$ &  the difference between correct and faulty left inputs of round $i$, where $ 0 \le i < T$. \\
\rowcolor[gray]{.9} $\delta^i_{j}$ &   the difference at bit $j$ of $\Delta^i$, where $ 0 \le i < T$, and $0 \le j < n$.\\
$S^a(X)$ & Circular left rotation of a $n$ bit word $X$ by $a$ bits.\\

\bottomrule
\end{tabular}

\end{table}

%\begin{table}[h]
%\label{tab:notations}
%\begin{tabular}{p{1.cm} c p{6.5cm}}
%$T$ & : & denotes the total number of rounds of the cipher.\\
%$(X^i, Y^i)$ &  : & denotes input of round $i$, $i=0,1,\ldots,T-1$. \\
%$(X^T, Y^T)$ &  : & denotes ciphertext. \\
%$(\hat{X}^{ i}, \hat{Y}^{ i})$ &  : & denotes faulty input of round $i$, $i=0,1,\ldots,T-1$. \\
%$(\hat{X}^{ T}, \hat{Y}^{ T})$ &  : & denotes faulty ciphertext. \\
%$X^i + Y^i$ & : & denotes 'XOR' operation of $X^i$ and $Y^i$.\\
%$X^i Y^i$ & : & denotes 'AND' operation of $X^i$ and $Y^i$.\\
%$(x^i_j, y^i_j)$ &  : & variables denoting bit $j$ of the input of the round $i$, $i=0,1,\ldots,T$, $ \quad j=0,1,\ldots,n-1$.\\
%$(\hat{x}^{i}_j, \hat{y}^{i}_j)$ &  : & variables denoting bit $j$ of the faulty input of the round $i$, $i=0,1,\ldots,T$, $ \quad j=0,1,\ldots,n-1$.\\
%$x^i_{j} + y^i_{j}$ & : & denotes bitwise 'XOR' operation of $x^i_j$ and $y^i_j$. \\
%$x^i_{j} y^i_{j}$ & : & denotes bitwise 'AND' operation of $x^i_j$ and $y^i_j$. \\
%$K^i$ & : & denotes the round key in round $i$, $i=0,1,\ldots,T-1$. \\
%$k^i_{j}$ &  : & denotes $j^{th}, 0 \le j \le n-1$ bit round key of round $i$.\\
%$\Delta^i$ & : & denotes the difference between correct and faulty left inputs of round $i$, where $ 0 \le i < T$. \\
%$\delta^i_{j}$ &  : & denotes the difference at bit $j$ of $\Delta^i$, where $ 0 \le i < T$, and $0 \le j < n$.\\
%$S^a(X)$ & : & Circular left rotation of a $n$ bit word $X$ by $a$ bits.\\
%\end{tabular}
%\end{table}

\subsection{SIMECK Specification}

The SIMECK family of lightweight block ciphers was introduced in CHES 2015~\cite{simeck} and was optimized for hardware implementations.
Similar to SIMON, SIMECK is based on a typical Feistel design and comprises three simple operations, namely, the bit-wise `AND', `rotation' and `XOR' operations. Let $n$ denote the word size. Then, SIMECK$2n/4n$ refers to perform encryptions or decryptions on $2n$-bit message blocks using a $4n$-bit key, where $n = 16, 24$ or $32$ is called the word size of SIMECK$2n/4n$.

\begin{figure}[ht]
\centering
\caption{One round of SIMECK block cipher}
\label{fig:Simeck}
\includegraphics[width=0.35\textwidth]{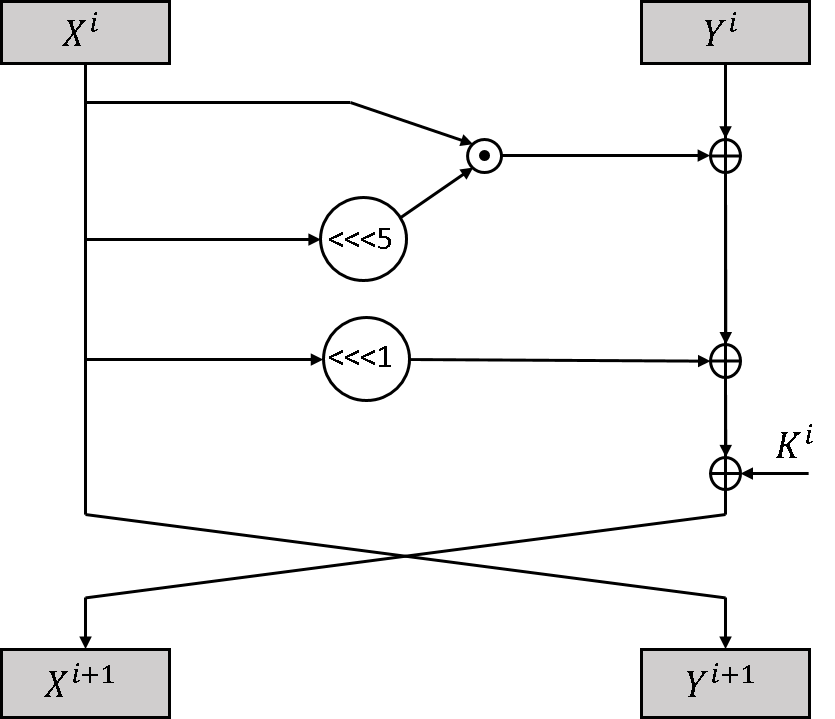}

\end{figure}

Fig.~\ref{fig:Simeck} shows a round function of SIMECK. Basically, the design of this family is based on the balanced Feistel network. There are in total $32, 36$ and $44$ encryption rounds for SIMECK32/64, SIMECK48/96, an SIMECK64/128, respectively. In the following, let us consider encrypting a plaintext $P$. At each round $i$, the input message is divided into two words $X^i$ and $Y^i$, where $P = X^0 || Y^0$ is the $2n$-bit plaintext. The round function of SIMECK is defined as follows:

\begin{equation}\label{eq:roundfunction}
(X^{i + 1}, Y^{i + 1}) = (Y^i + F(X^i) + K^i, X^i)
\end{equation}

\noindent where $F(X) = X \cdot S^5(X) + S^1(X)$, $K^i$ is the round key at the round $i$, and $S^a$ denotes a circular left rotation by $a$ bits. Let $(x^i_{n - 1}, \ldots, x^i_{0})$, and $(y^i_{n - 1}, \ldots, y^i_{0})$ denote the input of round $i$ for $0 \le i \le T-1$. The following relations hold for any $j = 0,\ldots, n - 1$:
\begin{align}
\label{eq:simeck}
&x^{i+1}_{j} = f(x^i_j) + y^i_{j} + k^{i}_{j}, \nonumber \\
&y^{i + 1}_{j} = x^i_{j},
\end{align}
where $f(x^i_j) = (x^i_{j} \; \& \; x^i_{(j - 5) \bmod n}) + x^i_{(j - 1) \bmod n}$. As all indices are computed modulo $n$, in what follows, $x^i_{j}$ will denote $x^i_{j \bmod n}$, $x^i_{j} x^i_{j'}$ will denote $x^i_{j \bmod n} \; \& \; x^i_{j' \bmod n}$, and $x^i_{j}  + x^i_{j'}$ will denote $x^i_{j \bmod n} \; + \; x^i_{j' \bmod n}$. \\

\noindent {\bf Key Schedule:} In SIMECK ciphers, the round key is generated from the master $K$ as follows. Firstly, the master is segmented into four words $K = (K^3, K^2, K^1, K^0)$, then for $i \ge 0$:

$$K^{i + 4} = K^i + F(K^{i + 1}) + C + (z_j)_i,$$

where $C$ is a constant value defined by $C = 2^n - 4$, $(z_j)_i$ denotes the $i$-th bit of the sequence $z_j$. SIMECK32/64 and SIMECK48/96 use the same sequence $z_0$, which can be generated by the primitive polynomial $X^5 + X^2 + 1$ with the initial state $(1, 1, 1, 1, 1)$, whereas SIMECK64/128 uses the sequence $z_1$, which can be generated by the primitive polynomial $X^6 + X + 1$ with the initial state $(1, 1, 1, 1, 1, 1)$.
Compared to SIMON ciphers, SIMECK ciphers are more efficient and compact hardware implementation due to its reuse of the round function with the round constant $C + (z_j)_i$ in the key scheduling algorithm.

\subsection{Security Analysis on SIMECK Ciphers}
Since SIMECK is based on SIMON and SPECK ciphers~\cite{BSS+13}, the security level of SIMECK ciphers is expected to be {\em comparable} to the those of SIMON ciphers.

Since their publication, there have been a number of academic works analyzing the security of SIMECK ciphers~\cite{bagheri2015linear, kolbl2016brief, qiao2016differential, qin2016linear, zhang2016integral, zhang2017security, sadeghi2018improved}. We briefly recall a few attacks in this section. Kolbl and Roy~\cite{kolbl2016brief} performed differential and linear cryptanalysis~\footnote{Differential cryptanalysis is a chosen-plaintext attack that studies how differences in input can affect the resultant difference at the output. Linear cryptanalysis is a known-plaintext attack in which the attacker studies probabilistic linear relations (called linear approximations) between parity bits of the plaintext, the ciphertext, and the secret key.}~\cite{biham1991differential, matsui1993linear}, the two most widely used attacks on block ciphers. They managed to break up to 19/32, 26/36, and 33/44 rounds of SIMECK32/64, SIMECK48/96, and SIMECK64/128 ciphers, respectively, with higher probability. These results cover more rounds compared to those against SIMON ciphers~(see \cite[Table 6]{kolbl2016brief}). Then, an improved differential attack using the dynamic key-guessing technique~\cite{qiao2016differential} was able to break up to 22, 28, and 35 rounds of SIMECK32/64, SIMECK48/96, and SIMECK64/128 ciphers. Even better, using the same technique with linear hull cryptanalysis, Qin et al.~\cite{qin2016linear} were able to break 1 or 2 more rounds compared to results in~\cite{qiao2016differential}.

\subsection{Differential Fault Analysis (DFA) on SIMECK Ciphers}

\subsubsection{Differential Fault Analysis}

The {\it Differential Fault Analysis} or {\it DFA} was first proposed by Biham and Shamir in~\cite{BS97}. Unlike the statistical attacks mentioned above such as differential and linear attacks, which requires a large amount of data to perform an attack, a DFA attack is able to recover the full master key by using just a few pair of plaintexts/ciphertexts.

In the Biham-Shamir's DFA attack against DES block cipher~\cite{DES}, an attacker will inject a fault in some rounds when executing a cryptographic implementation to obtain a faulty ciphertext. By analyzing the difference between the correct and faulty ciphertexts, the last round key could be revealed. Then, with the knowledge of the last round key, the attacker decrypts the correct ciphertext to obtain the input of the last round, which is the output of the second last round. After that, the attacker repeats this procedure to obtain more round keys until the master key can be deduced by the key schedule.
Subsequently, various DFA attacks on block ciphers have also been  carried out, including attacks against AES~\cite{Gir04}, Triple DES~\cite{HJQ04}, IDEA~\cite{CGV08}, SIMON and SPECK~\cite{tupsamudre2014differential, LYK19}, and SIMECK~\cite{NSS16}.

\subsubsection{Fault Models}

There exist various techniques to inject a fault into a computing device during its execution. Low-cost techniques include power glitch, clock tempering, or temperature variation~\cite{BDL97,KK99}. Kommerling and Kuhn reported that glitch attacks at the external power and clock supply lines are the most useful in practice~\cite{KK99}. 
High-cost techniques include light/laser injections or electromagnetic (EM) disturbances~\cite{SA02,SAS+02}.

%In these techniques, the attacker is able to insert transient faults leading to single bit errors. Since these techniques do not require precise timing, the faults tend to occur uniformly throughout the computation. 

These techniques allow specific control of a single register, a bus, or memory. 

The fault characteristics resulting from a fault injection are commonly captured in a {\it fault model}. A fault model will indicate the following characteristics: the location of the fault, the number of bits affected by the fault, and the fault types. 
In~\cite{riviere2015novel}, Riviere {\it et al.} classified 
fault models as follows:

\begin{itemize}
\item {\em Bit-wise models}: in which faults will manipulate a single bit. There are five types of bit-wise fault models: bit-set, bit-flip, bit-reset, stuck-at, and random-value.
\item {\em Byte-wise models}: in which faults will modify a byte at once. There are three types of byte-wise fault models: byte-set, byte-reset, or random-byte. 
\item {\em Wider models}: in which faults will manipulate an entire word that can be from 8 to 64 bits depending on the given architecture. 
\end{itemize}

Theoretical works often assume that the attacker has a precise control on both timing and location, {\em i.e.}, he knows the attacked bit as well as the attacked operation. 
In practice, it would be more challenging to inject a fault into a precise bit. 
In this paper, we consider the bit-flip model in which a single bit will be flipped to its complementary value, either from 0 to 1 or from 1 to 0. We also assume that the location of the flipped bit is {\em unknown} to the attacker.

\subsubsection{DFA attacks on SIMECK block cipher}
In the following, we briefly review differential fault analysis against SIMECK cipher in the {\it bit-flip} and {\em random byte} fault models. \\\\
In~\cite{NSS16}, Nalla \emph{et al.} presented the first differential fault analysis against SIMECK ciphers. Specifically, they demonstrated two DFA attacks. The former makes use of the {\it bit-flip fault model} and could recover the $n$-bit {\em last round key} using $n/2$ bit faults. The latter makes use of the {\it random byte fault model} that is more practical and could retrieve the last round key using $n/8$ faults. The process could be repeated four times to recover 4 round keys that are enough to recover the full master key.

Basically, their attacks mainly exploit the information leaked by the `AND' operation, which is the only non-linear function of SIMECK. Specifically, the attacker injects a fault in the intermediate left half ciphertext $X^{T-2}$, where $T$ is the number of rounds of SIMECK.  If one of the two input bits of the `AND' operation is $0$, then flipping the other input bit does not affect the output bit of $X^{T-1} = Y^T$. From the following equation, we can observe that one can deduce the last round key $K^{T-1}$ if the value $X^{T - 2}$ is known.

\begin{equation}\label{eq:lastroundkey}
K^{T-1} = X^{T-2} + F(Y^T) + X^T
\end{equation}

Suppose that the $j^{th}$ bit of $X^{T - 2}$ was flipped, one is able to deduce the value of $(j - 5)^{th}$ and $(j + 5)^{th}$ bits of $X^{T - 2}$, and thus recover the corresponding bits of $K^{T - 1}$ using Eq.~(\ref{eq:lastroundkey}). In other words, $2$ bits of $K^{T - 1}$ will be disclosed with every bit flipped in $X^{T-2}$. Therefore, by assuming that one could control the injected fault location, it requires $n/2$ faulty ciphertexts to retrieve the $n$-bit key. However, if the attacker has no control over the location of the flipped bit, the average number of faults required is approximately double, namely, around $n$ faults.\\\\

\section{Preliminary Observations and Facts}\label{sec:observations}
In this section,  some
 facts and observations that can be used to analyze faults on SIMECK block ciphers will be discussed. First, we consider the following lemma.

\begin{lemma}\label{lem:difference}
Let $X^t = \{x^t_0, x^t_1 \ldots, x^t_{n - 1}\}$ and $\hat{X}^t = \{\hat{x}^{t}_{0}, \hat{x}^{t}_{1}, \ldots, \hat{x}^{t}_{n - 1}\}$ be the correct and faulty left inputs respectively of the intermediate $t$-th round, $0 \le t < T$. Let $\delta^t_j = x^t_j + \hat{x}^{t}_{j}$, for $0 \le j < n $, be the differential representation of two correct and faulty bits $x^t_j$ and $\hat{x}^{t}_{j}$. We have:

\begin{dmath}\label{eq:difference}
\delta^{t + 1}_j = \delta^t_{j} x^t_{j - 5} + \delta^t_{j - 5}  x^t_{j} + \delta^t_{j} \delta^t_{j - 5} + \delta^t_{j - 1} + \delta^{t - 1}_j
\end{dmath}
\end{lemma}

\begin{proof}
From Eq.~(\ref{eq:simeck}), we have:
\begin{align*}
x^{t+1}_j &= x^t_{j} x^t_{j - 5} + x^t_{j - 1} + y^t_{j} + k^t_j, \quad \mbox{ and }\\
\hat{x}^{t+1}_j &= \hat{x}^{t}_{j} \hat{x}^{t}_{j - 5} + \hat{x}^{t}_{j - 1} + \hat{y}^{t}_{j} + k^t_{j}
\end{align*}

Note that $y^t_j = x^{t - 1}_j$. Then, summing up the two above equations gives:
$$\delta^{t + 1}_j = x^t_{j} x^t_{j - 5} + \hat{x}^{t}_{j} \hat{x}^{t}_{j - 5} + \delta^t_{j - 1} + \delta^{t - 1}_j$$

We have:
\begin{small}

\begin{align*}
&x^t_{j} x^t_{j - 5} + \hat{x}^{t}_j \hat{x}^{t}_{j - 5} \\&= (x^t_j + \hat{x}^{t}_j) (x^t_{j - 5} + \hat{x}^{t}_{j - 5}) + \hat{x}^{t}_j x^t_{j - 5} + x^{t}_j \hat{x}^{t}_{j - 5}\\&
 = \delta^t_j \delta^t_{j - 5} + (\hat{x}^{t}_j x^t_{j - 5}  + x^t_j x^t_{j -8}) + (x^{t}_j \hat{x}^{t}_{j - 5} + x^t_j x^t_{j -8}) \\&
= \delta^t_j \delta^t_{j - 5} + \delta^t_j x^t_{j - 5} + \delta^t_{j - 5}  x^t_j
\end{align*}

\end{small}

\noindent As a result, $\delta^{t + 1}_j = \delta^t_j x^t_{j - 5} + \delta^t_{j - 5}  x^t_j + \delta^t_j \delta^t_{j - 5} + \delta^t_{j - 1} + \delta^{t - 1}_j$. \qed

\end{proof}

Lemma \ref{lem:difference} tells us  that each bit $\delta^{t + 1}_j$ can be represented in terms of intermediate plaintext bits and input differences in the previous rounds.
More particular, the value of $\delta^{t + 1}_j$ depends on 2 bits of the intermediate input $X^t$, $3$ bits of the input difference $\Delta^t$, and one bit of $\Delta^{t - 1}$.
This allows us to construct a differential trail table to record and trace the $\delta^i_j$ values. Table~\ref{tab:trail_6rounds} shows the differential trail of SIMECK32/64 when we inject faults into the round $27$. The other trail tables are listed in  Fig.~\ref{sec:appendix}.

\begin{figure}[h]
\centering\footnotesize
\caption{Differential Trails of SIMECK Ciphers. Without loss of generality, we suppose that a fault will be injected into the bit $0$. The notation $*$ denotes for a complex non-linear expression of variables that is non-trivial to deduce their values.}\label{sec:appendix}
\vspace{0.2cm}
%\begin{tabular}{|p{8.27cm}|}
%\hline
%Let $\Delta^i = X^i + \hat{X}^{ i} = \{\delta^{i}_{j}\}_{0 \le i \le T,  0 \le j \le n-1}$ denote the difference of intermediate correct and faulty inputs at the round $i$. Without loss of generality, we suppose that a fault will be injected into the bit $0$. We list below differential trails of all different members of SIMECK lightweight block ciphers. Likewise, the notation $*$ denotes for a complex non-linear expression of variables that is non-trivial to deduce their values.\\
%\\
%\end{tabular}
%\begin{tabular}{|cp{7.3cm}|}
%\hline
%\multicolumn{2}{|c|}{SIMECK32/64}\\
%&\\
%$\Delta^{27}$ : & ($1$, $0$, $0$, $0$, $0$, $0$, $0$, $0$, $0$, $0$, $0$, $0$, $0$, $0$, $0$, $0$) \\
%$\Delta^{28}$ : & ($x^{27}_{11}$, $1$, $0$, $0$, $0$, $x^{27}_{5}$, $0$, $0$, $0$, $0$, $0$, $0$, $0$, $0$, $0$, $0$) \\
%$\Delta^{29}$ : & ($*$, $x^{28}_{12} + x^{27}_{11}$, $1$, $0$, $0$, $*$, $x^{28}_{6} + x^{27}_{5}$, $0$, $0$, $0$, $*$, $0$, $0$, $0$, $0$, $0$) \\
%$\Delta^{30}$ : & ($*$, $*$, $x^{29}_{13} + x^{28}_{12} + x^{27}_{11}$, $1$, $0$, $*$, $*$, $x^{29}_{7} + x^{28}_{6} + x^{27}_{5}$, $0$, $0$, $*$, $*$, $0$, $0$, $0$, $*$) \\
%$\Delta^{31}$ : & ($*$, $*$, $*$, $x^{30}_{14} + x^{29}_{13} + x^{28}_{12} + x^{27}_{11}$, $*$, $*$, $*$, $*$, $x^{30}_{8} + x^{29}_{7} + x^{28}_{6} + x^{27}_{5}$, $0$, $*$, $*$, $*$, $0$, $0$, $*$) \\
%$\Delta^{32}$ : & ($*$, $*$, $*$, $*$, $*$, $*$, $*$, $*$, $*$, $*$, $*$, $*$, $*$, $*$, $0$, $*$)\\
%&\\
%\end{tabular}

\begin{tabular}{|cp{7.5cm}|}
\hline
\multicolumn{2}{|c|}{SIMECK48/96}\\
&\\
$\Delta^{31}$ : & ($1$, $0$, $0$, $0$, $0$, $0$, $0$, $0$, $0$, $0$, $0$, $0$, $0$, $0$, $0$, $0$, $0$, $0$, $0$, $0$, $0$, $0$, $0$, $0$)\\

$\Delta^{32}$ : & ($x^{31}_{19}$, $1$, $0$, $0$, $0$, $x^{31}_{5}$, $0$, $0$, $0$, $0$, $0$, $0$, $0$, $0$, $0$, $0$, $0$, $0$, $0$, $0$, $0$, $0$, $0$, $0$)\\

$\Delta^{33}$ : & ($*$, $x^{32}_{20} + x^{31}_{19}$, $1$, $0$, $0$, $*$, $x^{32}_{6} + x^{31}_{5}$, $0$, $0$, $0$, $*$, $0$, $0$, $0$, $0$, $0$, $0$, $0$, $0$, $0$, $0$, $0$, $0$, $0$)\\

$\Delta^{34}$ : & ($*$, $*$, $x^{33}_{21} + x^{32}_{20} + x^{31}_{19}$, $1$, $0$, $*$, $*$, $x^{33}_{7} + x^{32}_{6} + x^{31}_{5}$, $0$, $0$, $*$, $*$, $0$, $0$, $0$, $*$, $0$, $0$, $0$, $0$, $0$, $0$, $0$, $0$)\\

$\Delta^{35}$ : & ($*$, $*$, $*$, $x^{34}_{22} + x^{33}_{21} + x^{32}_{20} + x^{31}_{19}$, $1$, $*$, $*$, $*$, $x^{34}_{8} + x^{33}_{7} + x^{32}_{6} + x^{31}_{5}$, $0$, $*$, $*$, $*$, $0$, $0$, $*$, $*$, $0$, $0$, $0$, $*$, $0$, $0$, $0$)\\

$\Delta^{36}$ : & ($*$, $*$, $*$, $*$, $x^{35}_{23} + x^{34}_{22} + x^{33}_{21} + x^{32}_{20} + x^{31}_{19}$, $*$, $*$, $*$, $*$, $x^{35}_{9} + x^{34}_{8} + x^{33}_{7} + x^{32}_{6} + x^{31}_{5}$, $*$, $*$, $*$, $*$, $0$, $*$, $*$, $*$, $0$, $0$, $*$, $*$, $0$, $0$)
\\
&\\
\end{tabular}
\begin{tabular}{|cp{7.5cm}|}
\multicolumn{2}{|c|}{SIMECK64/128}\\
&\\
$\Delta^{39}$ : & ($1$, $0$, $0$, $0$, $0$, $0$, $0$, $0$, $0$, $0$, $0$, $0$, $0$, $0$, $0$, $0$, $0$, $0$, $0$, $0$, $0$, $0$, $0$, $0$, $0$, $0$, $0$, $0$, $0$, $0$, $0$, $0$)\\

$\Delta^{40}$ : & ($x^{39}_{27}$, $1$, $0$, $0$, $0$, $x^{39}_{5}$, $0$, $0$, $0$, $0$, $0$, $0$, $0$, $0$, $0$, $0$, $0$, $0$, $0$, $0$, $0$, $0$, $0$, $0$, $0$, $0$, $0$, $0$, $0$, $0$, $0$, $0$)\\

$\Delta^{41}$ : & ($*$, $x^{40}_{28} + x^{39}_{27}$, $1$, $0$, $0$, $*$, $x^{40}_{6} + x^{39}_{5}$, $0$, $0$, $0$, $*$, $0$, $0$, $0$, $0$, $0$, $0$, $0$, $0$, $0$, $0$, $0$, $0$, $0$, $0$, $0$, $0$, $0$, $0$, $0$, $0$, $0$)\\

$\Delta^{42}$ : & ($*$, $*$, $x^{41}_{29} + x^{40}_{28} + x^{39}_{27}$, $1$, $0$, $*$, $*$, $x^{41}_{7} + x^{40}_{6} + x^{39}_{5}$, $0$, $0$, $*$, $*$, $0$, $0$, $0$, $*$, $0$, $0$, $0$, $0$, $0$, $0$, $0$, $0$, $0$, $0$, $0$, $0$, $0$, $0$, $0$, $0$)\\

$\Delta^{43}$ : & ($*$, $*$, $*$, $x^{42}_{30} + x^{41}_{29} + x^{40}_{28} + x^{39}_{27}$, $1$, $*$, $*$, $*$, $x^{42}_{8} + x^{41}_{7} + x^{40}_{6} + x^{39}_{5}$, $0$, $*$, $*$, $*$, $0$, $0$, $*$, $*$, $0$, $0$, $0$, $*$, $0$, $0$, $0$, $0$, $0$, $0$, $0$, $0$, $0$, $0$, $0$)\\

$\Delta^{44}$ : & ($*$, $*$, $*$, $*$, $x^{43}_{31} + x^{42}_{30} + x^{41}_{29} + x^{40}_{28} + x^{39}_{27}$, $*$, $*$, $*$, $*$, $x^{43}_{9} + x^{42}_{8} + x^{41}_{7} + x^{40}_{6} + x^{39}_{5}$, $*$, $*$, $*$, $*$, $0$, $*$, $*$, $*$, $0$, $0$, $*$, $*$, $0$, $0$, $0$, $*$, $0$, $0$, $0$, $0$, $0$, $0$) \\
\hline
\end{tabular}
\end{figure}

\begin{small}
\begin{table*}[ht]
\centering
\caption{Differential trail for the left half of the last 7 rounds  of SIMECK32/64 when flipping the least significant bit $0$ at the round $27$. The notation $*$ denotes a non-linear expression that would not be used in our attack.}
\label{tab:trail_6rounds}
\begin{tabular}{p{0.10\textwidth}p{0.10\textwidth}p{0.10\textwidth}p{0.10\textwidth}p{0.13\textwidth}p{0.10\textwidth}p{0.10\textwidth}p{0.07\textwidth}}
%\begin{tabular}{|c|c|c|c|c|c|c|c|}\hline
\toprule
\textbf{Bit} & ${\bf 0}$&${\bf 1}$&${\bf 2}$&${\bf 3}$&${\bf 4}$&${\bf 5}$&${\bf 6}$\\
\midrule \midrule

\rowcolor[gray]{.9} $\Delta^{26}$ &$0$&$0$&$0$&$0$&$0$&$0$&$0$\\
$\Delta^{27}$ &$1$&$0$&$0$&$0$&$0$&$0$&$0$\\

\rowcolor[gray]{.9} $\Delta^{28}$ &$x^{27}_{11}$ & $1$ &$0$ & $0$ & $0$ & $x^{27}_5$ &$0$ \\
$\Delta^{29}$ & $*$ & $x^{28}_{12} + x^{27}_{11}$ & $1$ & $0$ & $0$  & $*$ & $x^{28}_{6} + x^{27}_{5}$ \\

\rowcolor[gray]{.9} $\Delta^{30}$ & $*$ & $*$ & $x^{29}_{13} + x^{28}_{12} + x^{27}_{11}$ & $1$ & $0$ & $*$ & $*$ \\

$\Delta^{31}$ & $*$ & $*$ & $*$ & $x^{30}_{14} + x^{29}_{13} + x^{28}_{12} + x^{27}_{11}$ & $*$ & $*$ & $*$ \\

\rowcolor[gray]{.9} $\Delta^{32}$ & ~ Known values &&&&&&\\
%$\Delta^{T}$ &$x^{T - 5}_{6} + x^{T - 4}_{8} + x^{T - 3}_{10} + x^{T - 2}_{12} + \delta^{T - 1}_{14}$&$*$&$*$&$*$&$*$&$*$&$*$&$x^{T - 5}_{8} + x^{T - 4}_{10} + x^{T - 3}_{12} + x^{T - 2}_{14} + \delta^{T - 1}_{0}$\\ \hline
\bottomrule
\end{tabular}

%\rowcolor[gray]{.9} $\Delta^{32}$ &  \multicolumn{7}{c}{Known values}\\

\vspace{.3cm}
\begin{tabular}{p{0.10\textwidth}p{0.10\textwidth}p{0.12\textwidth}p{0.06\textwidth}p{0.06\textwidth}p{0.06\textwidth}p{0.06\textwidth}p{0.06\textwidth}p{0.06\textwidth}p{0.07\textwidth}}
%\begin{tabular}{|c|c|c|c|c|c|c|c|}\hline
\toprule
{\bf Bit}  & ${\bf 7}$ &${\bf 8}$ & ${\bf 9}$&${\bf 10}$&${\bf 11}$&${\bf 12}$&${\bf 13}$&${\bf 14}$&${\bf 15}$\\
\midrule \midrule

\rowcolor[gray]{.9}$\Delta^{26}$ & $0$ & $0$ & $0$ & $0$ & $0$ & $0$ & $0$ & $0$ & $0$\\
$\Delta^{27}$ & $0$ & $0$ & $0$ & $0$ & $0$ & $0$ & $0$ & $0$ & $0$\\

\rowcolor[gray]{.9}$\Delta^{28}$ & $0$ & $0$ & $0$& $0$& $0$& $0$& $0$& $0$& $0$ \\
$\Delta^{29}$ & $0$ & $0$ & $0$& $*$ & $0$ & $0$ & $0$ & $0$& $0$ \\
\rowcolor[gray]{.9}$\Delta^{30}$ & $x^{29}_{7} + x^{28}_{6} + x^{27}_{5}$ & $0$ & $0$ & $*$& $*$ & $0$ & $0$ & $0$ & $*$ \\

$\Delta^{31}$ & $*$ & $x^{30}_{8} + x^{29}_{7} + x^{28}_{6} + x^{27}_{5}$ & $0$ & $*$ & $*$ & $*$ & $0$ & $0$ & $*$ \\
\rowcolor[gray]{.9}$\Delta^{32}$ &  ~ Known values &&&&&&&&\\
%$\Delta^{T}$ & $*$&$*$&$*$&$*$&$*$&$*$&$*$&$*$\\ \hline
\bottomrule
\end{tabular}
\end{table*}

\end{small}

\begin{small}
\begin{table*}[htbp]
\footnotesize
\centering
\caption{Differential trail for the left half of the last 6 rounds  of SIMECK48/96 and SIMECK64/128 when a fault injected at the round $T - 5$ causing one bit flipped. The notation $*$ denotes an algebraic expression of immediate input bit variables.}
\label{tab:zeros}

%\begin{tabular}{|c|lll|}\hline

\begin{tabular}{p{0.09\textwidth}p{0.08\textwidth}p{0.08\textwidth}p{0.08\textwidth}}
\toprule
{SIMECK48/96} & \multicolumn{3}{c}{Left Input Differences (Bit position) } \\
\rowcolor[gray]{.9}\textbf{Round} & 0--7 & 8--15 & 16--23\\
\midrule \midrule
 $\Delta^{31}$ & 10000000 & 00000000 & 00000000 \\
\rowcolor[gray]{.9}$\Delta^{32}$ & *1000*00 & 00000000 & 00000000 \\
$\Delta^{33}$ & **100**0 & 00*00000 & 00000000 \\
\rowcolor[gray]{.9} $\Delta^{34}$ & ***10*** & 00**000* & 00000000 \\
 $\Delta^{35}$ & ****1*** & *0***00* & *000*000 \\
\rowcolor[gray]{.9}$\Delta^{36}$ & ******** & ******0* & **00**00 \\
\bottomrule
\end{tabular}
~~~~~~
%\begin{tabular}{|c|llll|}\hline
\begin{tabular}{p{0.09\textwidth}p{0.08\textwidth}p{0.08\textwidth}p{0.08\textwidth}p{0.08\textwidth}}
\toprule
SIMECK64/128 & \multicolumn{4}{c}{Left Input Differences (Bit position) } \\
\rowcolor[gray]{.9}\textbf{Round} & 0--7 & 8--15 & 16--23 & 24--31 \\
\midrule \midrule
 $\Delta^{39}$ & 10000000 & 00000000 & 00000000 & 00000000 \\
\rowcolor[gray]{.9}$\Delta^{40}$ & *1000*00 & 00000000 & 00000000 & 00000000 \\
 $\Delta^{41}$ & **100**0 & 00*00000 & 00000000 & 00000000  \\
\rowcolor[gray]{.9}$\Delta^{42}$ & ***10*** & 00**000* & 00000000 & 00000000 \\

$\Delta^{43}$ & ****1*** & *0***00* & *000*000 & 00000000 \\
\rowcolor[gray]{.9}  $\Delta^{44}$ & ******** & ******0* & **00**00 & 0*000000 \\
\bottomrule
\end{tabular}

\end{table*}
\end{small}

\begin{remark}[Input differences at the round $T - 2$.]\label{rem:deducedeltat2}
Let $T$ be the number of round of a SIMECK cipher, given the left input differences at the rounds $T - 1$ and $T$ ({\em i.e.}, $\Delta^{T - 1}$ and $\Delta^T$, resp.), and the left input of the round $T - 1$, {\em i.e.}, $X^{T - 1}$, then the left input differences of the round $T - 2$ could be computed as follows:

\begin{align}\label{eq:deduceDeltat2}
\Delta^{T - 2} &= \Delta^{T - 1} S^5(X^{T - 1}) + S^5(\Delta^{T - 1}) X^{T - 1} \\ \nonumber & + \Delta^{T - 1} S^5(\Delta^{T - 1}) +  S^1(\Delta^{T - 1}) +  \Delta^{T}
\end{align}
\end{remark}
 Eq.~(\ref{eq:deduceDeltat2}) can be easily obtained because from Eq.~(\ref{eq:difference}) we have:
\begin{equation*}\label{eq:deducedeltat2}
\delta^{T - 2}_j = \delta^{T - 1}_j x^{T - 1}_{j - 5} + \delta^{T - 1}_{j - 5}  x^{T - 1}_j + \delta^{T - 1}_j \delta^{T - 1}_{j - 5} + \delta^{T - 1}_{j - 1} + \delta^{T}_j
\end{equation*}

Once an attacker obtains the correct ciphertext $C = (X^T, Y^T)$ and the faulty ciphertext $\hat{C} = (\hat{X}^{ T}, \hat{Y}^{ T})$, s/he first computes values: $X^{T - 1} = Y^T$, $\hat{X}^{T - 1} = \hat{Y}^T$, $\Delta^{T} = X^T + \hat{X}^T$, and $\Delta^{T - 1} = X^{T - 1} + \hat{X}^{T - 1}$. Then, s/he is able to deduce the value of $\Delta^{T - 2}$ via Eq.~(\ref{eq:deduceDeltat2}).

\begin{observation}[Fault Propagation]
For each presence of fault in the round $t$ at the bit position $j$, {\em i.e.}, $\delta^t_j = 1$, it will affect into $3$ left input differences in the next round $t + 1$, they are values of $\delta^{t + 1}_{j}$, $\delta^{t + 1}_{j + 1}$, and $\delta^{t + 1}_{j + 5}$.
\end{observation}

From Lemma~\ref{lem:difference}, we have:
\begin{align*}
\delta^{t + 1}_j & = \delta^t_j x^t_{j - 5} + \delta^t_{j - 5}  x^t_j + \delta^t_j \delta^t_{j - 5} + \delta^t_{j - 1} + \delta^{t - 1}_j \\
\delta^{t + 1}_{j + 1} & = \delta^t_{j + 1} x^t_{j - 4} + \delta^t_{j - 4}  x^t_{j + 1} + \delta^t_{j + 1} \delta^t_{j - 4} + \delta^t_{j} + \delta^{t - 1}_{j + 1}  \\
\delta^{t + 1}_{j + 5} & = \delta^t_{j + 5} x^t_{j} + \delta^t_{j}  x^t_{j + 5} + \delta^t_{j + 5} \delta^t_{j} + \delta^t_{j + 4} + \delta^{t - 1}_{j + 5}.
\end{align*}

The values are dependent on the value of $\delta_{j}^{t}$. From Table~\ref{tab:trail_6rounds}, it can be seen that if one injects a fault at the bit position $0$ of the round $27$, that is, $\delta_0^{27} = 1$, then three bit positions $0, 1$ and $5$ at the round $28$ will be affected, they are $\delta^{28}_0$, $\delta^{28}_1$, and $\delta^{28}_5$. These faults continue propagating into the next round in the same way.

\begin{lemma}\label{lem:numberzero}
Suppose that a fault is injected into the left input of the round $t$ and one-bit flipped at the position $\ell$ is made, {\em i.e.}, $\delta^t_{\ell} = 1$ and $\delta^t_{j} = 0$ for $j \ne \ell$. Then, for $i \ge 1$,

\begin{enumerate}
\item $\delta^{t + i}_{\ell + i} = 1$ for $i \le \min\{4, \frac{n}{5}\}$
\item $\delta^{t + i}_{\ell + j} = 0$ for $1 \le i \le 3$, and $i < j \le 4$
\item $\delta^{t + i}_{\ell + j} = 0$ for $2 \le i \le 4$, and $i + 5 \le j \le i + 8$
\item $\delta^{t + i}_{j} = 0$ for $\ell + 5i < j < n + \ell$.
\end{enumerate}

\end{lemma}

\begin{proof}
This lemma can be easily proved due to Eq.~(\ref{eq:difference}). It can be seen that that the bit-flipped fault $\delta^{t}_{\ell}$ will affect to three input differences in the next round, that is, $\delta^{t + 1}_{\ell}$, $\delta^{t + 1}_{\ell + 1}$ and $\delta^{t + 1}_{\ell + 5}$. Bits at the position $j$, where $j \in [0, \ell)$ and $j \in (\ell + 5, n)$ will not be affected, {\em i.e.}, $\delta^{t + 1}_{j} = 0$. Furthermore, $\delta^{t + 1}_{\ell + 1} = \delta^t_{\ell} = 1$ and $\delta^{t + 1}_{\ell + j} = 0$, for $2 \le j \le 4$ due to Lemma~\ref{lem:difference}.
Likewise, the rightmost bit position $\ell + 5$ at the round $t + 1$ will propagate faults into three bits in the next round $t + 2$ ({\em i.e.}, $\delta^{t + 2}_{\ell + 5}$, $\delta^{t + 2}_{\ell + 6}$ and $\delta^{t + 2}_{\ell + 10}$) and $\delta^{t + 2}_{\ell + 2} = \delta^{t + 1}_{\ell + 1} = 1$.
The process will continue in the next round and so on.
%Furthermore, the total number of consecutive zeros will be $n - 5i - 1$ at the round $t + i$. As shown in Table~\ref{tab:zeros}, there are $8$ (resp., $16$) consecutive zeros in $\Delta^{34}$ (resp., $\Delta^{42}$) when one bit is flipped at the round $31$ (resp., $39$ of SIMECK48/96 (resp., SIMECK64/128).

\qed
\end{proof}

\begin{observation}\label{obs:deducefromzeros}
%If $\delta^{t - 1}_j$ is known according to Lemma~\ref{lem:numberzero} ({\em i.e.}, $\delta^{t - 1}_j = 0$), then
From Lemma~\ref{lem:difference}, we have:
\begin{align*}
\mbox{If } &\delta^t_j = 1 \,\, \& \,\, \delta^t_{j - 5} = 0,  \,\, \mbox{then }\, x^t_{j - 5} = \delta^{t + 1}_j + \delta^t_{j - 1} + \delta^{t - 1}_j\\
\mbox{If } &\delta^t_j = 0 \,\,\&  \,\, \delta^t_{j - 5} = 1,  \,\,  \mbox{then } x^t_{j} =  \delta^{t + 1}_j + \delta^t_{j - 1} + \delta^{t - 1}_j\\
\end{align*}
\end{observation}

\begin{table}[htbp]
\centering
\caption{Recover $x^t_{j - 5}$ and $x^t_{j}$ from Observation~\ref{obs:deducefromzeros}}\label{tab:recover_twobits}
%\begin{tabular}{ccccc}
\begin{tabular}{p{0.03\textwidth}p{0.03\textwidth}p{0.1\textwidth}p{0.1\textwidth}p{0.1\textwidth}}
\toprule
$\delta^t_j$ & $\delta^t_{j - 5}$ & $\delta^{t - 1}_j$ & $x^t_{j - 5}$ & $x^t_{j}$ \\
\midrule \midrule
\rowcolor[gray]{.9}0 & 0 & known/unknown & unknown & unknown\\
0 & 1 & known & unknown & $\delta^{t + 1}_j + \delta^t_{j - 1} + \delta^{t - 1}_j$\\
\rowcolor[gray]{.9}0 & 1 & unknown & unknown & unknown\\
1 & 0 & known & $\delta^{t + 1}_j + \delta^t_{j - 1} + \delta^{t - 1}_j$ & unknown \\
\rowcolor[gray]{.9}1 & 0 & unknown & unknown & unknown\\
1 & 1 & known/unknown & unknown & unknown\\
\bottomrule
\end{tabular}
\end{table}

\noindent Table~\ref{tab:recover_twobits} shows the possibility to recover two bits $x^t_{j - 5}$ and $x^t_{j}$ based on the relation of $\delta^t_j$ and $\delta^t_{j - 5}$. %Observation~\ref{obs:deducefromzeros}
%It indicates that we can recover some bits of $X^{T - 2}$ as we know $\Delta^{T - 1}$, $\Delta^{T - 2}$ and some bits of $\Delta^{T - 3}$ due to Lemma~\ref{lem:numberzero}.

\section{Differential Fault Analysis on SIMECK ciphers}\label{sec:bitflip}
This section will describe our DFA attack against SIMECK lightweight block ciphers in the bit-flip model. Specifically, there will be one bit flipped when a fault is injected. Given a plaintext $P$, the SIMECK encryption function outputs the corresponding ciphertext $C$. %Let $(X^i, Y^i)$ denote the intermediate inputs of $P$ at round $i$, for $0 \le i \le T - 1$.
Assume that a fault is injected into the input $X^t$ of an intermediate round $t$ and causes a bit-flip at the position $\ell$.

Let $\hat{X}^{t}$ be the fault value, so $X^t$ and $\hat{X}^{t}$ will be different at the $\ell^{th}$ bit and identical everywhere else. In other words, if $(x^t_{0}, \ldots, x^t_{n - 1})$ and $(\hat{x}^t_{0}, \ldots, \hat{x}^t_{n - 1})$ are $n$-bits of $X^t$ and $\hat{X}^t$, respectively, then $x^t_j = \hat{x}^{t}_j + 1$ for $j = \ell$ and $x^t_j = \hat{x}^{t}_j$ for $j \ne \ell$.

\subsection{Attack Description}\label{sec:attackdescription}

We aim at recovering the full master key $K$ (equivalent to $4$ round keys) by injecting faults into a single round only. Faults will be injected at the round $T - 5$ and we try to obtain 4 round keys $K^{T - 1}, K^{T - 2}, K^{T - 3}$ and $K^{T - 4}$. For instance, in order to retrieve $4$ round keys of SIMECK32/64 ($K^{28}$, $K^{29}$, $K^{30}$, and $K^{31}$), we will inject faults into the round $27$.
Our analytic attack based on the differential trail works as follows ({\sc Algorithm 1}):\\

%\vspace{1cm}

\par\noindent\rule{0.49\textwidth}{1.6pt}
\textbf{\sc Algorithm 1:} DFA attack on SIMECK ciphers
\par\noindent\rule{0.49\textwidth}{0.6pt}

\underline{\sc Step 1:} Choose a random plaintext $P$ to feed into the SIMECK encryption function, and get a ciphertext $C$ as return.

 \underline{\sc Step 2:} Re-run encryption with the above input $P$, and then inject a fault into the left input at the round $T - 5$.
Without loss of generality, we suppose that this fault flips the least significant bit $x^{T - 5}_{0}$.

\underline{\sc Step 3:} Find input differences of the subsequent rounds $T - 5 < t < T$, {\em i.e.}, $\Delta^t = X^t + \hat{X}^{ t}$. These input differences can be determined due to Lemma~\ref{lem:difference} and could be represented by $0, 1$ or algebraic expressions of other input variables. (as shown in Table~\ref{tab:trail_6rounds}).
The differences for SIMECK32/64 when the bit $x^{27}_0$ flipped are listed in Table~\ref{tab:trail_6rounds}. Such differences for SIMECK48/96 and SIMECK64/128 are listed in Fig~\ref{sec:appendix}. It can be seen that each round contains $2$ linear expressions to express the differences between correct and faulty intermediate ciphertexts.

\underline{\sc Step 4:} Deduce bits of $X^{T - 2}$,
\begin{itemize}
\item \underline{From linear expressions}:
As $\Delta^{T}, \Delta^{T - 1}$, and $\Delta^{T - 2}$ are known, the attacker can deduce two bits of $X^{T - 2}$ by summing up linear equations in two round $T - 2$ and $T - 1$. For example, if a fault was injected into the position $0$ at the round $27$ of SIMECK32/64, as shown in Table~\ref{tab:trail_6rounds}, the attacker can obtain the values of $x^{30}_{8}$ and $x^{30}_{14}$ by summing up $\delta^{31}_{8} + \delta^{30}_{7}$, and $\delta^{31}_{3} + \delta^{30}_{2}$, respectively. Let $\ell$ be the fault position at the round $T - 5$. As the fault will propagate to the position $\ell + 3$ at the round $T - 2$, the attacker will be able to deduce two bits $x^{T - 2}_{\ell - 2}$ and $x^{T - 2}_{\ell + 8}$.

\item \underline{From Observation~\ref{obs:deducefromzeros}}: The attacker knows $\Delta^{T - 1}$, $\Delta^{T - 2}$ and some bits of $\Delta^{T - 3}$ according to Lemma~\ref{lem:numberzero}. Thus, if $\delta^{T - 2}_j = 1$ and $\delta^{T - 3}_j = 0$ (resp., $\delta^{T - 3}_{j + 5} = 0$), then thanks to Observation~\ref{obs:deducefromzeros} the attacker can deduce two bits $x^{T - 2}_{j - 5}$ (and  $x^{T - 2}_{j + 5}$, resp.).
\end{itemize}

 \underline{\sc Step 5:} After obtaining bits of $X^{T - 2}$, the attacker can retrieve the corresponding bits of the last round key $K^{T - 1}$ due to Eq.~(\ref{eq:lastroundkey}). He/she repeats steps 2--4 with different flipped bit position to recover the full round key $K^{T - 1}$.

 \underline{\sc Step 6:} Decrypt the last round using $K^{T - 1}$ to get $(X^{T - 1}, Y^{T - 1}) (=(X^{T - 1}, X^{T - 2})) $ and $(\hat{X}^{ T - 1}, \hat{Y}^{ T - 1}) = ((\hat{X}^{ T - 1}, \hat{X}^{ T - 2}))$. Again, as the values $\Delta^{T - 1}$, $\Delta^{T - 2}$, and $X^{T - 2}$ are known, the attacker could compute $\Delta^{T - 3}$ due to Eq.~(\ref{eq:difference}). He/she then repeats steps 4 and 5 to obtain the round key $K^{T - 2}$. The attacker continues this process to retrieve two more round keys $K^{T - 3}$ and $K^{T - 4}$.

 \underline{\sc Step 7:} With $4$ round keys, the attacker could deduce the master key $K$ of a SIMECK block cipher by the key schedule.
\par\noindent\rule{0.49\textwidth}{1.6pt}

From {\sc Step 4}, it can be seen that for each bit $1$ of $\Delta^{T - 2}$, the attacker may recover $2$ bits of $X^{T - 2}$ (corresponding to two bits of $K^{T - 1}$). Likewise, once $K^{T - 1}$ is recovered, that attacker knows $\Delta^{T - 2}$, $\Delta^{T - 3}$ and some bits of $\Delta^{T - 4}$, he then may be able to recover two bits of $X^{T - 3}$ (corresponding to two bits of $K^{T - 2}$) with each bit $1$ of $\Delta^{T - 3}$ and so on.

\begin{table*}[t]
%\scriptsize
\small
\centering
\caption{Comparisons of the Fault analysis attacks on SIMECK ciphers}%. *: the number of faults reported in~\cite{NSS16} is required to recover only one round key.}
\label{tab:comparison}
\begin{tabular}{l|cc|ccc}
\toprule

\multirow{3}{*}{SIMECK$2n/4n$} & \multicolumn{2}{c|}{Nalla \emph{et al.}~\cite{NSS16}} & \multicolumn{3}{c}{Our work} \\\cline{2-6}
& Locations of faults to & No of faults * & Locations of faults to & No of faults & No of faults   \\%\cline{2-6}
& recover master key & (Last round key) & recover master key & (Last round key) & (Master key) \\
\midrule \midrule
\rowcolor[gray]{0.9} SIMECK32/64 & $30, 29, 28, 27$ & $28.32$ & $27$ & $12.12$ & $25.78$ \\

SIMECK48/96 & $34, 33, 32, 31$ & $41.44$ & $31$ & $22.88$ & $43.56$  \\

\rowcolor[gray]{.9} SIMECK64/128 & $42, 41, 40, 39$ & $57.06$ & $39$  & $30.14$ & $89.51$ \\
\bottomrule

\end{tabular}

\end{table*}

%%%%%%%%%%%%%%%%%%%%%%%%%%%%%%%%%%%
\subsection{Fault Position}\label{sec:faultposition}
In this section, we will discuss  how the attacker determines the position of a fault.

\subsubsection{Determining fault position in SIMECK64/128.} A fault will be injected into the round $39$ at the position $\ell$, that is, $\delta^{39}_{\ell} = 1$. An attacker will use on the following facts to deduce the value of $\ell$:
\begin{itemize}
\item From Lemma~\ref{lem:numberzero}, the input differences at the round $42$ will have $16$ consecutive zeros (see Table~\ref{tab:zeros}).
\item The fault position $\ell$  will be propagated and shifted right three positions at the round $42$, that is $\delta^{42}_{\ell + 3} = 1$.
\item There is a pattern $10***00**000* \underbrace{00\ldots0}_{16 \, digits}$, where the value of $*$ could be 1 or 0. The position of the first $1$ in this pattern will be $\ell + 3$. This pattern consists of $29$ digits out of $32$ digits of $\Delta^{42}$. %could be unique with a high probability.
\end{itemize}

\subsubsection{Determining fault position in SIMECK48/96.} A fault will be injected into the round $31$ at the position $\ell$, that is, $\delta^{31}_{\ell} = 1$. Likewise, an attacker will use on the following facts to deduce the value of $\ell$:
\begin{itemize}
\item From Lemma~\ref{lem:numberzero}, the input differences at the round $34$ will have $8$ consecutive zeros (see Appendix~\ref{sec:appendix}).
\item The fault position $\ell$  will be propagated and shifted right three positions at the round $34$, that is $\delta^{34}_{\ell + 3} = 1$.
\item There is a pattern $10***00**000* \underbrace{00\ldots0}_{8 \, digits}$, where the value of $*$ could be 1 or 0. The position of the first $1$ in this pattern will be $\ell + 3$. This pattern consists of $21$ digits out of $24$ digits of $\Delta^{34}$.
\end{itemize}

\subsubsection{Determining fault position in SIMECK32/64.} As there is no consecutive bit zero determined by Lemma~\ref{lem:numberzero}, determining the fault position for SIMEC 32/64 is challenging. However, we still have this pattern $10***00**000$ in $16$ bits of $\Delta^{30}$, where the value of $*$ could be 1 or 0. The first $1$ in this pattern will be $\ell + 3$, where $j$ is the fault position injected at the round $27$. %This pattern could help to deduce the position $j$ of the fault with a high confidence.

Based on the above facts and given that $\Delta^{T - 2}$ is known, the attacker also can deduce the fault position $\ell$ at the round $T - 5$ with high confidence.

%%%%%%%%%%%%%%%%%%%%%%%%%%%%%%%%%%%%%%%%%%%%
\section{Experiments}\label{sec:simulation}

\subsection{Setup}
To verify our proposed attack, we implemented a software simulation against all three SIMECK family members in C programming language\footnote{The source code could be found in the link \url{https://github.com/dple/DFA_Simeck}}. 
In this simulation, we suppose that the attacker cannot control the position of the faults. He thus may inject faults into the same position many times until fully recovering one round key or the master key. Once a {\em random} fault is injected into the left input of the round $i$ at the bit $j$, only the differential input $\delta_j^i$ will be equal to one, differential inputs at other positions will be zero. The positions of the faults were generated {\em randomly} and {\em independently}.

For each cipher member, we repeat the experiment 10,000 times and report the average number of faulty ciphers required to recover the last round key $K^{T - 1}$ and the whole master key (corresponding to the last 4 round keys $K^{T - 1}$, $K^{T - 2}$, $K^{T - 3}$, and $K^{T - 4}$).  The only purpose of recovering the last key is to compare our attacks to the ones in~\cite{NSS16}.

\subsection{Simulation Results}

\begin{figure*}
\centering
\caption{Histogram of the number of faults required to recover keys. The number of samples is 10,000. Frequency density is represented on the vertical axis and the number of fault injections is represented on the horizontal axis.}\label{fig:histogram}
\includegraphics[width=0.75\textwidth]{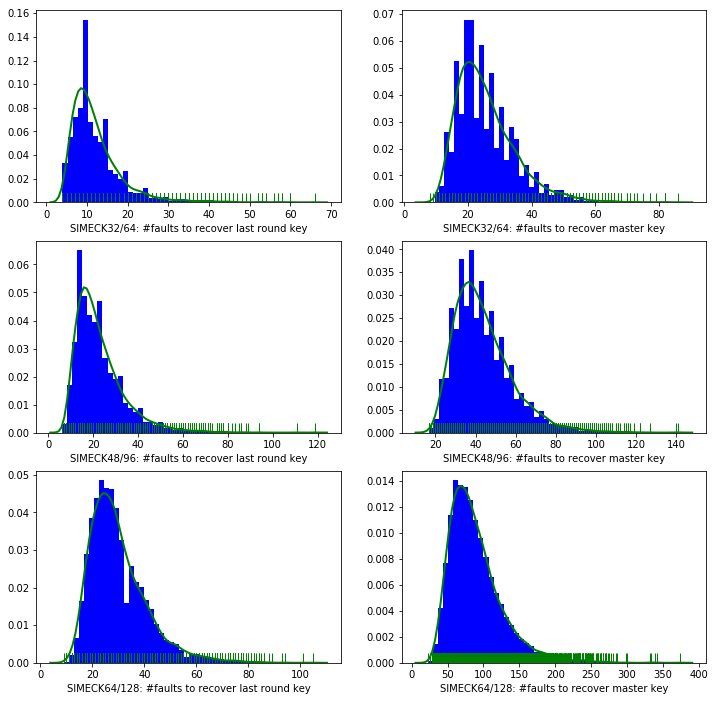}
\end{figure*}

Fig.~\ref{fig:histogram} shows a histogram of the number of fault injections to obtain the last round key and the full master key of all three variants of SIMECK block cipher. 
%Frequency is represented on the vertical axis and the number of fault injections is represented on the horizontal axis. 
The horizontal axis represents the number of faults, and the vertical axis represents the frequency experiments requiring that number of faults. 
The total number of experiments is 10,000, as mentioned in the previous section. As seen in Fig.~\ref{fig:histogram}, all histograms are right-skewed. In order to recover the last round key in SIMECK32/64, it requires from a few fault injections to about 70 fault injections. However, most of the experiments only required around 9-13 fault injections. Likewise, the maximum number of fault injections required to recover the full master key is about 90; however, most of the experiments only required around 19-25 fault injections.

Table~\ref{tab:comparison} compares the round locations, in which faults are induced, the number of faults required between our attacks, and Nalla \emph{et al.}'s attack under the one bit-flip model. 
From Table~\ref{tab:comparison}, it can be seen that our fault analysis attack requires much fewer faults than Nalla \emph{et al.}'s attack to recover the last round key for all three members of SIMECK. Even more, our attack to recover the whole master key of SIMECK32/64 also requires fewer faults than theirs for only the last round key. Last but not least, in order to recover the whole master key,  while our attack injects the faults into a single round only, Nalla \emph{et al.}'s attack has to inject faults into $4$ different rounds of a SIMECK cipher. As a result, our attack is more {\em practical}.

\subsection{Discussions}
The proposed attack is based on an analytic method, which analyzes the differential trail between the correct and faulty ciphertexts to deduce the intermediate inputs, and then round keys. The key point in this method is to find the differential trail that is not too complicated to analyze. Compared to theoretical attacks~\cite{bagheri2015linear, qiao2016differential, zhang2016integral, zhang2017security}, this method is more straightforward, but we demonstrated that the attack is {\em effective}. In this paper, we use only linear expressions to analyze, however, non-linear expressions (e.g., quadratic expressions denoted as $*$ in Table~2 and Fig.~2) may be useful for more-in-depth analyses. 
Block ciphers with more complicated round functions would be resistant to the proposed attack, e.g., using a non-linear S-box. However, this will trade-off with the performance of the implementation.

%%%%%%%%%%%%%%%%%%%%%%%%%%%%%%%%%%%%%%%%%%%%%%%%%%%%%%%%%%%%%%

\section{Conclusion}\label{sec:conclusion}
In this paper, we have proposed  an improved fault attack on the SIMECK lightweight block ciphers under the one bit-flip model. In this model, we assume that a single bit will be flipped to its complementary value once a fault was successfully injected. We also assume that the attacker has no control over the location of the flipped bit. 
The advantage of our attack is that not only it requires less number of faults, but also faults need to be injected into only a {\em single round} of the ciphers in order to recover the whole master key. 
As a result, it makes our attack  more {\em practical}. We demonstrated the effectiveness of our attack by simulating in C for all three members of SIMECK ciphers. % 
Our experimental results over $10,000$ times showed that the attack requires $25.78$, $43.56$, and $89.51$ faults on average to recover the full master key of the SIMECK32/64, SIMECK48/96, and SIMECK64/128, respectively.

\bibliographystyle{IEEEtran}
\bibliography{IEEEabrv,fa}

\begin{IEEEbiography}
[{\includegraphics[width=1in,height=1.25in,clip,keepaspectratio]{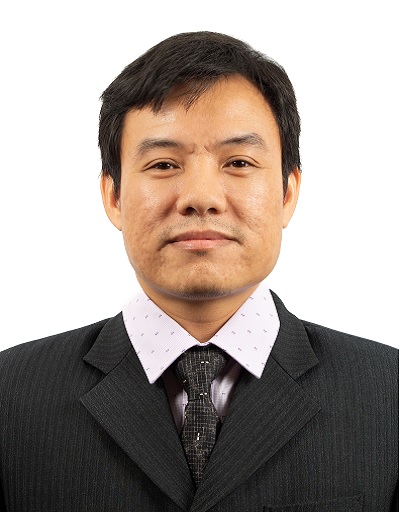}}]{Duc-Phong Le} received the Ph.D. degree in computer science from the University of Pau et des Pays de l’Adour, in August 2009. He was a Postdoctoral Fellow with the Algorithms Research Group, University of Caen, Base Normandie, from 2009 to 2010, a Research Scientist with the National University of Singapore (NUS), Singapore, from October 2010 to May 2016, a Senior Security Analyst with Underwriters Laboratories, from May 2016 to June 2017, and the Scientist II of the Institute for Infocomm Research (I2R), Agency of Science, Technology and Research (A*STAR), Singapore, from July 2017 to March 2019. He is currently working as the Research Team Lead of the Canadian Institute for Cybersecurity (CIC), University of New Brunswick, Canada. His research interests include applied cryptography, elliptic curve cryptography, secure and efficient implementations, applied machine learning to cybersecurity issues, and blockchain.
\end{IEEEbiography}

\begin{IEEEbiography}[{\includegraphics[width=1in,height=1.25in,clip,keepaspectratio]{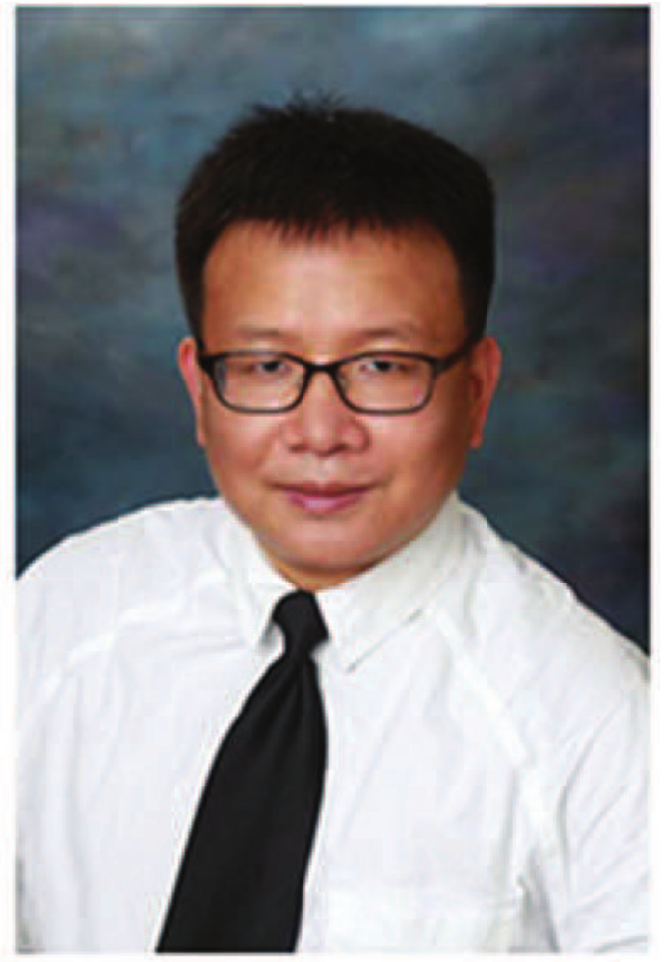}}]{Rongxing Lu}
 (S'09-M'11-SM'15) is currently an associate professor at the Faculty of Computer Science (FCS), University of New Brunswick (UNB), Canada. Before that, he worked as an assistant professor at the School of Electrical and Electronic Engineering, Nanyang Technological University (NTU), Singapore from April 2013 to August 2016. Rongxing Lu worked as a Postdoctoral Fellow at the University of Waterloo from May 2012 to April 2013. He was awarded the most prestigious ``Governor General's Gold Medal", when he received his PhD degree from the Department of Electrical \& Computer Engineering, University of Waterloo, Canada, in 2012; and won the 8th IEEE Communications Society (ComSoc) Asia Pacific (AP) Outstanding Young Researcher Award, in 2013. He is presently a senior member of IEEE Communications Society. His research interests include applied cryptography, privacy enhancing technologies, and IoT-Big Data security and privacy. He has published extensively in his areas of expertise, and was the recipient of 8 best (student) paper awards from some reputable journals and conferences. Currently, Dr. Lu currently serves as the Vice-Chair (Conferences) of IEEE ComSoc CIS-TC (Communications and Information Security Technical Committee). Dr. Lu is the Winner of 2016-17 Excellence in Teaching Award, FCS, UNB.
\end{IEEEbiography}

\begin{IEEEbiography}[{\includegraphics[width=1in,height=1.25in,clip,keepaspectratio]{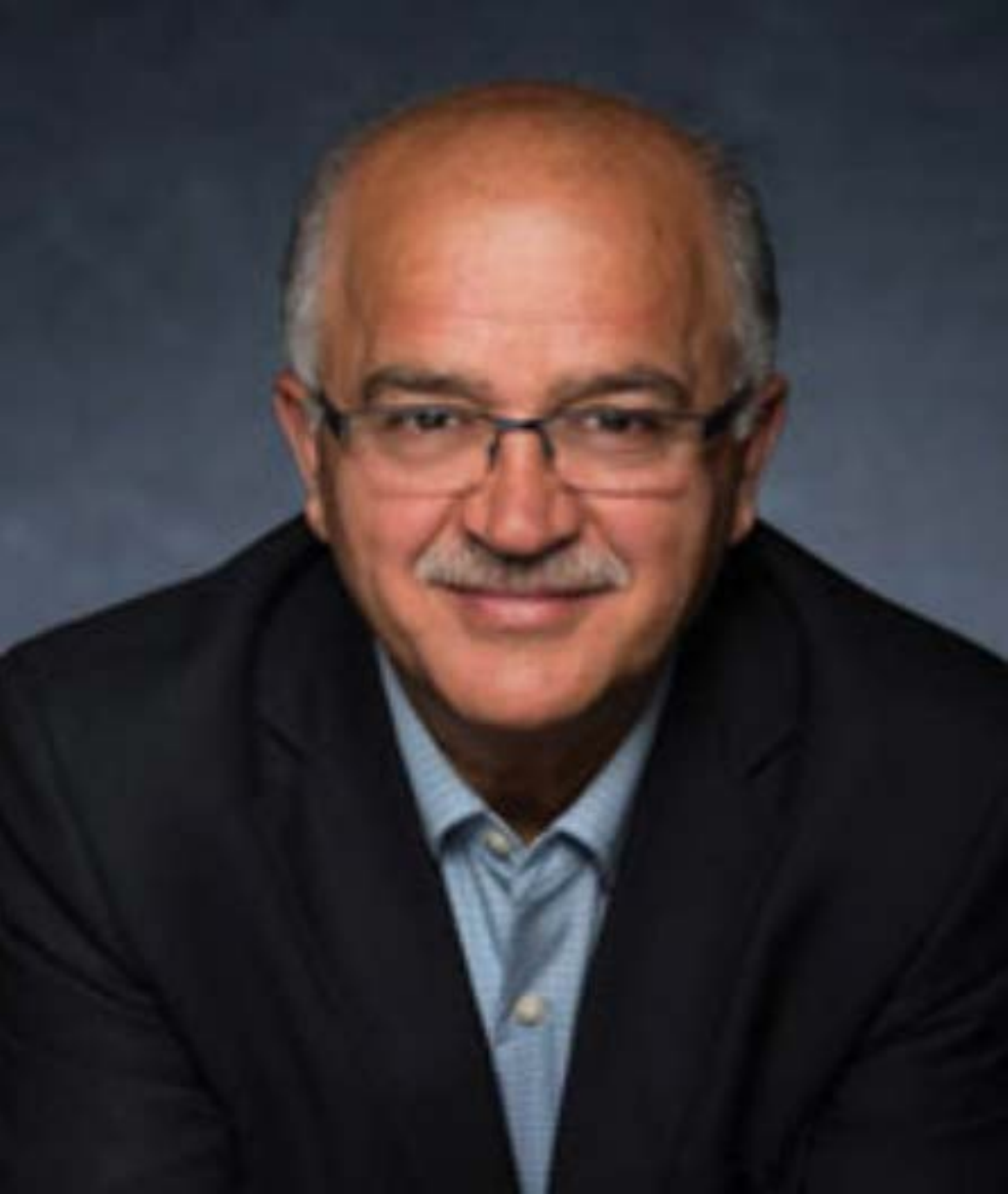}}]{Ali A. Ghorbani} (Senior Member, IEEE) has held a variety of academic positions for the past 39 years. He is currently a Professor of computer science, the Tier 1 Canada Research Chair in cybersecurity, and the Director of the Canadian Institute for Cybersecurity, which he established, in 2016. He has served as the Dean of the Faculty of Computer Science, University of New Brunswick, from 2008 to 2017. He is also the Founding Director of the Laboratory for Intelligence and Adaptive Systems Research. He has spent over 29 years of his 39-year academic career, carrying out fundamental and applied research in machine learning, cybersecurity, and critical infrastructure protection. He is the Co-Inventor on three awarded and one filed patent in the fields of cybersecurity and web intelligence. He has published over 280 peer-reviewed articles during his career. He has supervised over 190 research associates, postdoctoral fellows, and students during his career. His book Intrusion Detection and Prevention Systems: Concepts and Techniques (Springer, October 2010). He developed several technologies adopted by high-tech companies and co-founded three startups, Sentrant Security, EyesOver Technologies, and Cydarien Security, in 2013, 2015, and 2019, respectively. He was a recipient of the 2017 Startup Canada Senior Entrepreneur Award, and the Canadian Immigrant Magazine’s RBC top 25 Canadian immigrants of 2019. He is the Co-Founder of the Privacy, Security, Trust (PST) Network in Canada and its annual international conference. He has served as the Co-Editor-in-Chief for the International Journal of Computational Intelligence, from 2007 to 2017.
\end{IEEEbiography}

%\end{document}
\end{document}